\documentclass[copyright,creativecommons]{eptcs}
\providecommand{\event}{ACT 2026} 

\usepackage[utf8]{inputenc} 
\usepackage[T1]{fontenc}    

\usepackage[x11names,dvipsnames, svgnames,table]{xcolor}

\DeclareSymbolFont{cmlargesymbols}{OMX}{cmex}{m}{n}
\let\coprod\undefined
\DeclareMathSymbol{\coprod}{\mathop}{cmlargesymbols}{"60}

\newenvironment{colblock}[1]{%
  \color{#1}%
}{%
}

\newcommand{\highlightcol}{blue}

\usepackage{hyperref}
\hypersetup{
  colorlinks   = true,    
  urlcolor     = \highlightcol,    
  linkcolor    = \highlightcol,    
  citecolor    = \highlightcol     
}

\usepackage{amsmath}
\usepackage{amsthm}
\usepackage{thmtools}
\usepackage{hyperref}
\usepackage[capitalize]{cleveref}

\usepackage{url}            
\usepackage{booktabs}       
\usepackage{amsfonts}       
\usepackage{nicefrac}       
\usepackage{fancyhdr}       
\usepackage{graphicx}       

\usepackage{mathrsfs}
\usepackage{mathtools}
\usepackage{stmaryrd}
\usepackage{tikz}
\usepackage{quiver}
\usepackage{thmtools}

\usepackage{iftex}

\usepackage{verbatim}

\usepackage{microtype}
\usepackage{silence}
\ifpdf
  \usepackage{underscore}         
  \usepackage[T1]{fontenc}        
\else
  \usepackage{breakurl}           
\fi

\newtheorem{theorem}{Theorem}

\newtheorem{proposition}[theorem]{Proposition}
\newtheorem{lemma}[theorem]{Lemma}

\theoremstyle{definition}
\newtheorem{definition}[theorem]{Definition}
\newtheorem{example}[theorem]{Example}
\newtheorem{remark}[theorem]{Remark}

\newtheorem{construction}[theorem]{Construction}

\newcommand{\id}{\mathsf{id}}

\newcommand{\supp}{\operatorname{supp}}

\DeclarePairedDelimiterX{\infdivx}[2]{(}{)}{%
  #1\;\delimsize\|\;#2%
}

\newcommand{\namedCat}[1]{\textup{\textbf{#1}}}

\newcommand{\namedFunctor}[1]{\mathrm{#1}}

\newcommand{\CatC}{\namedCat{C}}

\newcommand{\Set}{\namedCat{Set}}

\newcommand{\FuncF}{\namedFunctor{F}}

\newcommand{\FuncT}{\namedFunctor{T}}
\newcommand{\unit}{\eta}
\newcommand{\mult}{\mu}

\newcommand{\dist}{d}                       
\newcommand{\Dist}{\namedFunctor{\MakeUppercase{\dist}}}   
\newcommand{\power}{p}
\newcommand{\Power}{\mathcal{\MakeUppercase{\power}}}

\newcommand{\PowerFinite}{\Power_{\operatorname{fin}}}  

\newcommand{\transitionCoalgFinal}{\omega}

\newcommand{\States}{S}
\newcommand{\states}{\MakeLowercase{\States}}

\newcommand{\Observations}{O}
\newcommand{\observations}{\MakeLowercase{\Observations}}

\newcommand{\inputs}{i}
\newcommand{\Inputs}{I}

\newcommand{\transitionCoalg}{f}

\newcommand{\transitionCoalgTr}{\mathsf{tr}}
\newcommand{\transitionCoalgOut}{\mathsf{out}}

\newcommand{\moore}{\mathsf{Moore}}

\newcommand{\structmoore}{{\FuncF \FuncT \moore}}

\newcommand{\structstochmoore}{{\FuncF \Dist \moore}}

\newcommand{\mealy}{\mathsf{Mealy}}

\newcommand{\stochmealy}{\mathsf{\Dist \mealy}}

\newcommand{\structmealy}{{\FuncT \mealy}}

\newcommand{\suppmealy}{{\FuncT \mathsf{Supp} \mealy}}

\newcommand{\stochsuppmealy}{{\Dist \mathsf{Supp} \mealy}}

\newcommand{\structunifmealy}{{\mathsf{Unif} \structmealy}}

\newcommand{\stochunifmealy}{{\mathsf{Unif} \stochmealy}}

\DeclareMathOperator{\Caus}{Caus}

\title{Bayesian updates from coalgebraic determinisation}
\author{Manuel Baltieri
\institute{Araya Inc., Japan\thanks{This work was funded by the Advanced Research + Invention Agency (ARIA) through project code MSAI-SE01-P011.}
\xdef\sharedthanks{\the\value{footnote}}}
\email{manuel\_baltieri@araya.org}
\and
Nathaniel Virgo
\institute{University of Hertfordshire, UK\footnotemark[\sharedthanks]}
\email{n.virgo@herts.ac.uk}
}
\def\titlerunning{Bayesian updates from determinisation}
\def\authorrunning{M. Baltieri \& N. Virgo}
\begin{document}
\frenchspacing
\maketitle

\begin{abstract}
    The powerset construction is the classical determinisation procedure for nondeterministic finite automata. 
    In the coalgebraic setting, this construction has been generalised to structured coalgebras, which are coalgebras equipped with extra data. 
    For stochastic Moore machines over the distribution monad, a type of structured coalgebra, the determinisation construction induces a semantics assigning to each finite input word a distribution on the current output.
    This semantics is appropriate when only the current output matters, but it is too coarse for settings in which intermediate observations must also be taken into account, as is typical for agents solving POMDPs in control theory and reinforcement learning. 
    In these contexts, agents need to condition on all realised observations, not just the final one, so to better plan for the future.
    This has been addressed from a category theoretic perspective through a procedure called \emph{unifilarisation}, which (in our context) takes a stochastic Mealy machine and produces a machine whose states are priors over the original state space and whose transitions are given by Bayesian filtering.
    Here we show that unifilarisation is an instance of coalgebraic determinisation. 
    We work with Mealy machines over monads equipped with extra structure generalising the notion of the support of a distribution. 
    We show that in this setting, unifilarisation arises from the general determinisation procedure. 
    We then compare the resulting final coalgebra semantics with the Moore-style one. 
    Instead of assigning only a distribution on current outputs to each finite input word, it yields causal stochastic behaviours, that is, families mapping input words to distributions on output words compatible with the ``causality'' constraint that outputs cannot depend on future inputs.
\end{abstract}

\section{Introduction}
Determinisation is classically introduced in automata theory as the move from nondeterministic finite automata to deterministic finite automata via the powerset (or subset) construction~\cite{rabinFiniteAutomataTheir1959}: given a nondeterministic finite automata with state space $\States$ and input space $\Inputs$, the determinised machine has state space $\PowerFinite(\States)$, with a transition function mapping a subset of states $\States' \in \PowerFinite(\States)$ and an input symbol $\inputs \in \Inputs$ to the set of all $\inputs$-successors of states in $\States'$, i.e. all the states that can be reached from $\States'$ after consuming an input symbol $\inputs$.
This construction yields an equivalent deterministic automaton that accepts the same language of finite words.

This construction has previously been generalised to a coalgebraic setup in~\cite{silvaGeneralizingPowersetConstruction2010, silvaGeneralizingDeterminizationAutomata2013, jacobsTraceSemanticsDeterminization2015, bonchiDistributionBisimilarityPower2021}.
Specifically, the construction in~\cite{silvaGeneralizingPowersetConstruction2010, silvaGeneralizingDeterminizationAutomata2013} starts with a \emph{structured coalgebra}, which is a coalgebra $(\States, f : \States \to \FuncF \FuncT(\States))$ of $\FuncF\FuncT$, where $\FuncT$ is a monad, together with some extra data.
The monad $\FuncT$ is used to model nondeterminism or stochasticity, with the precise type of nondeterminism (or ``branching'') determined by the monad.
The \emph{determinisation} procedure then turns a structured coalgebra into an algebra $\left( \FuncT(\States), \transitionCoalg^\sharp : \FuncT(\States) \to \FuncF \FuncT(\States) \right)$, which is a coalgebra of $\FuncF$ (and can therefore be seen as deterministic), whose carrier is $\FuncT(\States)$.
The traditional powerset construction is a special case of this framework, where the transition type is that of a Moore machine of a specific kind and the branching is given by the finite powerset monad $\FuncT = \PowerFinite$.
For this monad, determinisation groups states (by union) that can be reached by consuming the same input.
Its extension to (Rabin) probabilistic automata follows a similar pattern, but with the distribution monad with finite support $\FuncT = \Dist$.
In this case, determinisation averages transitions by a given distribution of initial states, i.e. it performs a pushforward of probabilities along the transitions of the given coalgebra.
One consequence of this construction is that, when final coalgebras exist, determinisation induces a behavioural semantics on determinised coalgebras. 
For the examples above, this semantics is indexed by finite input words. 
In particular, when $\FuncT = \Dist$, the determinised coalgebra lives on $\Dist(\States)$, and for each $n \ge 0$ assigns to every input word $w \in \Inputs^n$ a distribution in $\Dist(\Observations)$, obtained by propagating an initial distribution on states through the transitions induced by $w$. 
Thus the semantics records the distribution of the current output after a finite input history has been processed.

This semantics is natural when only the current output matters. 
However, there are settings in which one also needs, for each possible realised output, the corresponding update of the state distribution, so that prediction or control can be implemented consistently along an observed trajectory. 
In the probabilistic case, a different approach again replaces a system with state space $\States$ by a new one with state space $\Dist(\States)$, now interpreted as a space of \emph{beliefs} updated by Bayesian filtering~\cite{kalmanNewApproachLinear1960, jazwinskiStochasticProcessesFiltering1970}.
This is strikingly similar in spirit to determinisation of automata, though not obviously the same.
In control and reinforcement learning this yields the \emph{belief MDP} where posterior distributions over hidden states form a fully observable state space for prediction and decision making, and their transition is given by the Bayes filter~\cite{astromOptimalControlMarkov1965, striebelSufficientStatisticsOptimum1965, kaelblingPlanningActingPartially1998, subramanianApproximateInformationState2022}. 
Concretely, given a prior $p \in \Dist(\States)$ and an input $\inputs \in \Inputs$, one forms a joint distribution on $\Observations \times \States$ of next output and next state, then takes its marginal on $\Observations$, and for each possible output $\observations$ obtains a posterior distribution on next states conditioned on observing $\observations$.
The same construction appears in computational mechanics as the \emph{mixed-state presentation}~\cite{ellisonPredictionRetrodictionAmount2009, crutchfieldExactComplexitySpectral2016}, see~\cite[Section 5.1]{rosasAIVatFundamental2025}. 
There, finite observation histories induce distributions over the (hidden) states of a stochastic generator, called mixed states. 
This treatment highlights an important aspect of this process, \emph{unifilarity}: once the current mixed state and the realised observation are fixed, the next mixed state is uniquely determined. 
This has been formulated in category-theoretic terms as \emph{unifilarisation}, with a functor from categories of stochastic machines to categories of unifilar machines~\cite{virgoUnifilarMachinesAdjoint2023}.

In this paper, we show that unifilarisation is itself an instance of coalgebraic determinisation, but not of the standard Moore-style form. 
For structured Moore machines over the distribution monad $\Dist$, determinisation aggregates successor behaviour by averaging and yields final semantics of the form $\Inputs^\star \to \Dist(\Observations)$, recording only the distribution of the current output after a finite input history has been processed. 
By contrast, we work with Mealy machines and what we call \emph{monads with support}, that is, monads $\FuncT$ equipped with a function $\supp: \FuncT(X) \to \PowerFinite^+(X)$ for each $X$ and isomorphisms of the form
$
    \FuncT(X \times Y) \cong \coprod_{p \in \FuncT(X)} \FuncT(Y)^{\supp(p)}
$
for each pair of sets $X$ and $Y$, natural in $X$.
From these we define \emph{supported Mealy machines}, which are coalgebras of the functor
$
    \FuncF(\States) = (\coprod_{p \in \FuncT(\Observations)} \States^{\supp(p)})^\Inputs.
$
Supported Mealy machines then fit \cite{silvaGeneralizingPowersetConstruction2010, silvaGeneralizingDeterminizationAutomata2013}'s structured-coalgebra pattern as coalgebras
$
    f : \States \to \FuncF(\FuncT \States),
$
and we construct a compatible $\FuncT$-algebra
$
    \xi : \FuncT(\FuncF(\FuncT \States)) \to \FuncF(\FuncT \States).
$
The general determinisation procedure therefore sends a machine with state space $\States$ to one with state space $\FuncT(\States)$. 
For $\FuncT = \Dist$, the resulting determinised coalgebra is exactly the unifilarisation of a stochastic Mealy machine: given a prior and an input, it forms the prior-weighted joint law of next output and next state, then rewrites it as an output marginal together with posterior next-state distributions. 
These posterior distributions are precisely the same of belief MDPs and mixed-state presentations. 
This also induces a different final coalgebra semantics from the Moore case: instead of assigning only a distribution on current outputs to each finite input word, it yields causal stochastic behaviours, represented by families $b_n : \Inputs^n \to \Dist(\Observations^n)$ compatible with Bayesian filtering updates.

Belief MDPs have previously been approached from a coalgebraic perspective by~\cite{koriCoalgebraicDeterminizationBelief2026}, introducing a construction combining the existing coalgebraic determinisation with a new ``belief decomposition'' so to obtain objects called ``belief coalgebras'' that abstract belief MDPs.
Our work differs in that rather than introducing a new construction to accommodate `beliefs', we show that they arise naturally from the existing determinisation construction, once the context is set up correctly.

\section{Structured coalgebras}
This section provides an overview of the generalised (coalgebraic) determinisation construction approach introduced by~\cite{silvaGeneralizingPowersetConstruction2010, silvaGeneralizingDeterminizationAutomata2013}, focusing on a few core examples relevant for a comparison of final coalgebra semantics introduced for different kinds of determinations that can be shown to be special cases of this general construction.
Following~\cite{silvaGeneralizingPowersetConstruction2010, silvaGeneralizingDeterminizationAutomata2013}, we introduce structured coalgebras, or $\FuncF \FuncT$-coalgebras, coalgebras parametrised by a functor $\FuncF$ and a monad $\FuncT$, and carrying a $\FuncT$-algebra structure.
\begin{definition}[Structured coalgebras, or $\FuncF \FuncT$-coalgebras~\cite{silvaGeneralizingPowersetConstruction2010, silvaGeneralizingDeterminizationAutomata2013}]
    \label{def:structuredCoalgebras}
    Let $\CatC$ be a category, $\FuncF: \CatC \to \CatC$ an endofunctor and $(\FuncT, \unit, \mult)$ a monad on $\CatC$.
    An $\FuncF \FuncT$-coalgebra is a coalgebra $(\States, \transitionCoalg: \States \to \FuncF \FuncT(\States))$ of the endofunctor $\FuncF \FuncT$, together with a morphism $\xi:\FuncT\FuncF \FuncT(\States)\to \FuncF \FuncT(\States)$ such that $(\FuncF \FuncT(\States), \xi)$ is a $\FuncT$-algebra over the monad $(\FuncT, \unit, \mult)$.
\end{definition}
Explicitly, this means a structured coalgebra consists of:
\begin{itemize}
    \item an object $\States$ in $\CatC$,
    \item a morphism $f:\States \to \FuncF \FuncT(\States)$, to be seen as a coalgebra of $\FuncF \FuncT$,
    \item a morphism $\xi : \FuncT \FuncF \FuncT(\States) \to \FuncF \FuncT(\States)$,
\end{itemize}
such that the following diagrams, representing unit and action properties, commute:
\begin{equation}
    \begin{tikzcd}
        \FuncF \FuncT(\States) && {\FuncT \FuncF \FuncT(\States)} &&& {\FuncT \FuncT \FuncF \FuncT(\States)} && {\FuncT \FuncF \FuncT(\States)} \\
        \\
        && \FuncF \FuncT(\States) &&& {\FuncT \FuncF \FuncT(\States)} && \FuncF \FuncT(\States)
        \arrow["{\unit_{\FuncF \FuncT(\States)}}", from=1-1, to=1-3]
        \arrow["{\id_{\FuncF \FuncT(\States)}}"', from=1-1, to=3-3]
        \arrow["\xi", from=1-3, to=3-3]
        \arrow["{\FuncT \xi}", from=1-6, to=1-8]
        \arrow["\mult_{\FuncF \FuncT(\States)}"', from=1-6, to=3-6]
        \arrow["\xi", from=1-8, to=3-8]
        \arrow["\xi"', from=3-6, to=3-8]
    \end{tikzcd}
\end{equation}
The idea is that the monad $\FuncT$ models some kind of nondeterminism, and a coalgebra of $\FuncF \FuncT$ models a process with a nondeterministic component.
Typical choices for $\FuncT$ are $\FuncT = \PowerFinite$ for finite possibilities, $\FuncT = \Dist$ for discrete probabilities, or $\FuncT = ({-}) + 1$ for partiality.
An $\FuncF \FuncT$-coalgebra can seen as a system where the next-step behaviour is ``$\FuncF$-shaped'', and the successor states come with branching or nondeterminism of type $\FuncT$.
The coalgebraic determinisation procedure of~\cite{silvaGeneralizingPowersetConstruction2010, silvaGeneralizingDeterminizationAutomata2013} takes a structured coalgebra $(\States, \transitionCoalg, \xi)$ and produces a map $\transitionCoalg^\sharp: \FuncT(\States) \to \FuncF \FuncT(\States)$, to be seen as a coalgebra of $\FuncF$ with carrier $\FuncT(\States)$.
This is a deterministic system in the sense that it is a coalgebra of $\FuncF$ only, without the nondeterminism.

\begin{construction}[Determinisation of $\FuncF \FuncT$-coalgebras~\cite{silvaGeneralizingPowersetConstruction2010, silvaGeneralizingDeterminizationAutomata2013}]
    \label{cst:determinisation}
    A structured coalgebra $(\States, \transitionCoalg, \xi)$
    can be extended uniquely to a homomorphism of $\FuncT$-algebras $\transitionCoalg^\sharp : (\FuncT \States, \mu_\States) \to (\FuncF (\FuncT \States), \xi)$.
    This unique extension is given by
    \begin{align}
        \FuncT(\States) \xrightarrow{\transitionCoalg^\sharp} \FuncF \FuncT(\States)= \FuncT(\States) \xrightarrow{T f} \FuncT \FuncF \FuncT(\States) \xrightarrow{\xi} \FuncF \FuncT(\States).
    \end{align}
    The resulting map $\transitionCoalg^\sharp$ is a coalgebra of $\FuncF$ with carrier $\FuncT(\States)$ and is called the \emph{determinisation} of the coalgebra $\transitionCoalg$, or the \emph{determinised} coalgebra.
\end{construction}
Intuitively, $\xi$ tells us how to combine a $\FuncT$-collection of one-step behaviours into a single one-step behaviour.
Using this, the original coalgebra extends uniquely to $\transitionCoalg^\sharp \coloneqq \xi \circ \FuncT \transitionCoalg$, which gives the transitions of the determinised coalgebra.
In~\cite{silvaGeneralizingPowersetConstruction2010, silvaGeneralizingDeterminizationAutomata2013} it is shown that various constructions, including the classical powerset construction and the totalisation of partial automata, are instances of this.
We will later show that unifilarisation, or the construction of a ``belief machine'' is also an instance, for an appropriate choice of $\FuncF$, $\FuncT$ and $\xi$.

Here we are following~\cite{silvaGeneralizingPowersetConstruction2010, silvaGeneralizingDeterminizationAutomata2013} a single map $\xi$ turning a single structured coalgebra into a $\FuncT$-algebra, rather than a natural family of such maps, turning \emph{all} $\FuncF \FuncT$-coalgebras into $\FuncT$-algebras. 
The latter is discussed in more detail in~\cite{jacobsTraceSemanticsDeterminization2015} by using an (Eilenberg-Moore) distributive law. 
However, we won't explore this in the current paper.

\subsection{Example - Structured Moore machines}
\label{sec:structured-moore-machines-example}
Various examples of this parametrisation are given in~\cite{silvaGeneralizingPowersetConstruction2010, silvaGeneralizingDeterminizationAutomata2013}, including partial Mealy automata, nondeterministic automata, pushdown automata, etc. 
Here we are especially interested in one of these examples, structured Moore machines, which are structured coalgebras for the functor $\FuncF({-}) = \FuncT(\Observations) \times ({-})^\Inputs$ and an arbitrary monad $\FuncT$.
We are interested in these machines because they model stochastic systems with inputs an outputs, which are similar to the POMDPs of interest in control theory and reinforcement learning.
However, as we will see, their determinisation doesn't closely resemble the construction of a belief MDP.
We will see in \cref{sec:detMealyMachines} that the belief machine does arise when we make different choices of $\FuncF$, $\FuncT$ and $\xi$.
By comparing the two cases we will get some insight into why we should expect Bayesian updating to arise in one case but not the other.

To define structured Moore machines we need the following proposition, which is part of the standard theory of monads on $\Set$.
It relies on the fact that every monad on a Cartesian closed category has a canonical strength.

\begin{proposition}
    \label{prop:algebra-on-exponental}
    Given a monad $(\FuncT, \unit, \mult)$ on a Cartesian closed category, an algebra $(X, h)$ for $\FuncT$ and an arbitrary set $\Inputs$, the set $X^\Inputs$ also carries an algebra structure $h^{(\Inputs)} : \FuncT(X^\Inputs) \to X^\Inputs$ given by currying the map
    \begin{align}
        \label{eqn:currying-algebra}
        \FuncT(X^\Inputs) \times \Inputs
            \xrightarrow{\mathrm{str}_{X^\Inputs, \Inputs}}
        \FuncT(X^\Inputs \times \Inputs)
            \xrightarrow{T(\mathrm{eval}_{X, \Inputs})}
        \FuncT(X)
            \xrightarrow{h}
        X.
    \end{align}
\end{proposition}

\begin{proof}
    Omitted.
\end{proof}

\begin{definition}[Structured Moore machines~\cite{silvaGeneralizingDeterminizationAutomata2013}]
    \label{ex:tAlgebraStructMooreMachine}
    Given a monad $(\FuncT, \unit, \mult)$ and sets $\Inputs$ and $\Observations$, a $\FuncT$-structured Moore machine is an $\FuncF \FuncT$-coalgebra $(\States,\transitionCoalg_\structmoore, \xi)$, where $\FuncF = \FuncT(\Observations)\times ({-})^\Inputs$, and $\xi : \FuncT(\FuncT(\Observations) \times \FuncT(\States)^\Inputs) \to \FuncT(\Observations) \times \FuncT(\States)^\Inputs$ is given by
    \begin{align}
        \xi \coloneqq 
        \FuncT(\FuncT(\Observations) \times \FuncT(\States)^\Inputs)
        \xrightarrow{\langle \FuncT \pi_{\FuncT(\Observations)}, \FuncT \pi_{\FuncT(\States)^\Inputs} \rangle}
        \FuncT(\FuncT(\Observations)) \times \FuncT(\FuncT(\States)^\Inputs)
        \xrightarrow{\mult_{\Observations}\times \mult_S^{(\Inputs)}}
        \FuncT(\Observations) \times \FuncT(\States)^\Inputs
    \end{align}
    where $\pi_\inputs$ are projections for the $i$'th component of the Cartesian product and $\mult_S^{(\Inputs)}$ is \cref{prop:algebra-on-exponental} applied to the free algebra $\FuncT(\States)$.
\end{definition}
We will sometimes refer to the map $\transitionCoalg_\structmoore:\States\to\FuncT(\Observations)\times\FuncT(\States)$ by its components, 
\begin{align}
    \transitionCoalg_\structmoore = \langle \transitionCoalgOut_\structmoore : \States \to \FuncT(\Observations) , \; \transitionCoalgTr_\structmoore : \States \to \FuncT(\States)^\Inputs \rangle.
\end{align}
For particular choices of $\FuncT$, we can use this definition to obtain nondeterministic Moore machines (for the finite nonempty powerset monad, $\PowerFinite$) and stochastic Moore machines (for the distribution monad, $\Dist$).
Taking the distribution monad as an example, we can regard a structured Moore machine as a machine that, on each time step, produces an output in $\Observations$ stochastically, as well as taking in an input from $\Inputs$ and then stochastically updating its state.
Since structured Moore machines are structured coalgebras, they can be determinised using~\cref{cst:determinisation}, yielding the following.

\begin{proposition}[Determinised structured Moore machines~\cite{silvaGeneralizingDeterminizationAutomata2013}]
    \label{ex:determinisedStructMoore}
    Let $(\States, \transitionCoalg_\structmoore : \States \to \FuncT(\Observations) \times \FuncT(\States)^\Inputs)$ be a structured Moore machine, understood as an $\FuncF \FuncT$-coalgebra.
    By determinising this structured Moore machine (\cref{cst:determinisation}), we obtain a coalgebra $(\FuncT(\States), \transitionCoalg_\structmoore^\sharp : \FuncT(\States) \to \FuncT(\Observations) \times \FuncT(\States)^\Inputs)$,
    with $\transitionCoalg_\structmoore^\sharp$ given by 
    \begin{align}
        \transitionCoalg^\sharp_\structmoore = \langle \transitionCoalgOut^\sharp_\structmoore : \FuncT(\States) \to \FuncT(\Observations) , \; \transitionCoalgTr^\sharp_\structmoore : \FuncT(\States) \to \FuncT(\States)^\Inputs \rangle,
    \end{align}
    where
    \begin{align}
        \transitionCoalgOut^\sharp_\structmoore =
        \FuncT(\States)\xrightarrow{\FuncT(\transitionCoalg_\structmoore)}
        \FuncT(\FuncT(\Observations) \times \FuncT(\States)^\Inputs)
        \xrightarrow{\FuncT(\pi_{\FuncT(\Observations)})}
        \FuncT(\FuncT(\Observations))
        \xrightarrow{\mult_\Observations}
        \FuncT(\Observations)
    \end{align}
    and    
    \begin{align}
        \transitionCoalgTr^\sharp_\structmoore =
        \FuncT(\States)\xrightarrow{\FuncT(\transitionCoalg_\structmoore)}
        \FuncT(\FuncT(\Observations) \times \FuncT(\States)^\Inputs)
        \xrightarrow{\FuncT(\pi_{\FuncT(\States)^\Inputs})}
        \FuncT(\FuncT(\States)^\Inputs)
        \xrightarrow{\mult_\States^{(\Inputs)}}
        \FuncT(\States)^{\Inputs}.
    \end{align}
\end{proposition}
\begin{proof}
    This follows directly from applying~\cref{cst:determinisation} to coalgebras of the form $(\States, \transitionCoalg_\structmoore : \States \to \FuncT(\Observations) \times \FuncT(\States)^\Inputs)$.
\end{proof}
In the case of the distribution monad, $\FuncT = \Dist$, this produces a machine whose state space is $\Dist(\States)$.
The output map $\transitionCoalgOut^\sharp_\structstochmoore$ takes as input a measure over $\States$ and pushes it forward along the Markov kernel $S\xrightarrow{f_{\structstochmoore}}\Dist(\Dist(\Observations))\times\Dist(\States)^\Inputs)\xrightarrow{\Dist(\pi_{\Dist(\Observations)})}\Dist(\Dist(\Observations))
\xrightarrow{\mult_{\Observations}}\Dist(\Observations)$.
The transition map $\transitionCoalgTr_\structstochmoore^\sharp$ does something similar, pushing forward a measure over $\States$ along the map $\Dist(\pi_{\Dist(\States)^\Inputs})\circ \transitionCoalg^\sharp_\structstochmoore$ to get a measure over functions $\Inputs\to \Dist(\States)$ (that is, an element of $\Dist(\Dist(\States)^\Inputs)$). The $\mult_\States^{(\Inputs)}$ step can be seen as taking an input in $\Inputs$ and passing it to each of these functions to get an element of $\Dist(\Dist(\States))$, and then averaging to get an element of $\Dist(\States)$.
But this doesn't involve conditioning on an observation, since no sample from the output distribution is considered.

From the perspective of someone interested in reinforcement learning and POMDPs however, it might seem surprising that these dynamics involve pushing forward a measure along a Markov kernel but that they don't involve a Bayesian conditioning step.
The reason for this becomes clearer when considering the final coalgebra semantics.
A final coalgebra is said to capture the behaviour of a coalgebra~\cite{ruttenUniversalCoalgebraTheory2000, jacobsIntroductionCoalgebraMathematics2017}.
More precisely, the elements of a final coalgebra (when it exists) are the possible observable behaviours of all objects of a given category of coalgebras.

Following standard arguments from automata theory,~\cite{silvaGeneralizingPowersetConstruction2010, silvaGeneralizingDeterminizationAutomata2013} argue that for $\FuncF \FuncT$-coalgebras, we should consider behaviours to be the final coalgebra of $\FuncF$ rather than $\FuncF \FuncT$, so that the behaviours of a $\FuncF \FuncT$-coalgebra are the $\FuncF$-behaviours of its determinisation.
This is because the final coalgebra of $\FuncF \FuncT$ is too fine-grained: it can consider two machines to have different behaviours even if they are indistinguishable to an external observer, because they differ in the behaviour of their internal state.
Intuitively, determinisation `absorbs' these internal differences.

We now consider the final coalgebra semantics in this sense for determinised stochastic Moore machines, which are structured Moore machines with $\FuncT = \Dist$. The functor $\FuncF$ in this case is $\Dist(\Observations) \times({-})^\Inputs$.

\begin{proposition}[$\FuncF$-behaviours for stochastic Moore machines~\cite{silvaGeneralizingPowersetConstruction2010, silvaGeneralizingDeterminizationAutomata2013}]
    The endofunctor $\Dist(\Observations) \times ({-})^\Inputs$ has final coalgebra:
    \begin{align}
        \left( \Dist(\Observations)^{\Inputs^\star}, \transitionCoalgFinal_\structstochmoore^\sharp : \Dist(\Observations)^{\Inputs^\star} \to \Dist(\Observations) \times \left( \Dist(\Observations)^{\Inputs^\star} \right)^\Inputs \right).
    \end{align}
\end{proposition}
\begin{proof}
    This is obtained by applying~\cite[Proposition 2.3.5]{jacobsIntroductionCoalgebraMathematics2017} with the set $B = \Dist(\Observations)$.
\end{proof}

Elements of the carrier of this final coalgebra are maps $g : \Inputs^\star \to \Dist(\Observations)$ assigning to each finite input word $w \in \Inputs^\star$ a distribution $g(w) \in \Dist(\Observations)$ on observations. 
Equivalently, for each $\observations \in \Observations$, the quantity $g(w)(\observations) \in [0, 1]$ is the probability of observing $\observations$ after reading $w$. 
This generalises the usual semantics of (Rabin) probabilistic automata~\cite{rabinProbabilisticAutomata1963}: when $\Observations = 2$, each distribution $g(w) \in \Dist(2)$ is determined by a single scalar in $[0, 1]$, namely the probability assigned to the designated accepting output, so one recovers the standard acceptance probability of the word $w$.
We note that the output here is stochastic, but it only depends on the sequence of previous inputs.
This makes sense in terms of Rabin probabilistic automata, where the idea is to give a word as input to the machine, and obtain as output the probability of accepting or rejecting the word.
Since we only ever consider one output in this situation there is no need to invoke conditioning via Bayes, since there are no observations to condition on, besides the last output.

However, in different situations involving agents solving, e.g. POMDPs in reinforcement learning and control theory, one is interested not only in the final output given a sequence of inputs, but in the whole sequence of intermediate outputs that the machine emits as well.
This is because a controller can observe these intermediate outputs from the POMDP and might choose to give different inputs depending on what they are.
For this reason, we desire different semantics in a POMDP setting.
It turns out that this can be achieved in the determinisation framework of~\cite{silvaGeneralizingDeterminizationAutomata2013}, but we must choose a different class of automata.
This will be the topic of the rest of the paper.

\section{Monads with support}
For the purpose of this paper, and particularly for the instantiation of unifilarisation in coalgebraic terms following~\cite{virgoUnifilarMachinesAdjoint2023}, 
we will from now on to work with Mealy machines rather than Moore machines.
We will however need Mealy machines of a particular kind, so before defining them formally, we introduce the notion of monad with support and a few examples of it.

Monads with support are monads $\FuncT$ that come with some extra structure, which we call \emph{support}, because in the case of the distribution monad it corresponds to the support of a distribution.
This will be important because for monads with support we will obtain the finitely supported version of the definition of unifilarisation given in~\cite{virgoUnifilarMachinesAdjoint2023}, which will then also be derived from the coalgebraic determinisation construction of~\cite{silvaGeneralizingDeterminizationAutomata2013}.

As in the case of the coalgebraic determinisation of~\cite{silvaGeneralizingPowersetConstruction2010, silvaGeneralizingDeterminizationAutomata2013}, we restrict our attention here to monads on $\Set$.
This means that we do not attempt to extend~\cref{def:monadSupport} to measure-theoretic treatments.
The notion of monad with support will probably not directly transfer to such contexts, since the notion of support is quite subtle in category-theoretic probability~\cite{fritzAbsoluteContinuitySupports2026}, although see~\cref{rmk:relationships}.
\begin{definition}[Monad with support]
    \label{def:monadSupport}
    A \emph{monad with support} consists of a monad $T$ on $\mathbf{Set}$ together with, for each $X$ and each $p\in \FuncT(X)$, a subset $\supp(p) \subseteq X$ called the \emph{support} of $p$, and for each $X, Y$ an isomorphism 
    \begin{align}
    \label{eqn:monad-with-support-iso}
        \Phi_{X, Y} : \FuncT(X \times Y) \xrightarrow{\cong} \coprod_{p \in \FuncT(X)} \FuncT(Y)^{\supp(p)},
    \end{align}
    natural in $Y$.
\end{definition}

The reason for requiring naturality in $Y$ but not $X$ is that there is no obvious way to make the right-hand side of \cref{eqn:monad-with-support-iso} functorial in $X$.

The definition of $\supp$ depends on the specific monad.
We give some examples below, namely the distribution monad and the finite nonempty powerset monad.
Examples that don't fit the definition (at least not in the most obvious way) include the powerset monad (finite or not), the partiality monad and the subdistribution monad.

\begin{remark}[Relationship to hyper normalisation and strongly representable Markov categories]
    \label{rmk:relationships}
    In \cite{jacobs2017hyper}, conditional probability is formulated in terms of a what amounts to a map $\Dist(X \times Y)\xrightarrow{h_{X,Y}} \Dist(X \times \Dist(Y))$ called \emph{hyper normalisation}, where $D$ is the distribution monad.
    \cite{jacobs2017hyper} requires $X$ to be finite, but this is not necessary in order to define such a map.
    An intuition is that if we start with a joint distribution between $X$ and $Y$ and observe $X$, we will be left with a sample $x$ from $X$ and the conditional distribution over $Y$ given $x$.
    Since $x$ is an outcome of a random variable, the conditional distribution over $Y$ is also a random variable.
    This explains the domain of the hyper normalisation map, $\Dist(X \times \Dist(Y))$: it is a joint distribution between samples and the conditional distributions given those samples.
    
    A hyper normalisation-like map can be defined for any monad with support on $\mathbf{Set}$.
    It is given by the composite
    \begin{equation}
        \FuncT(X\times Y)
        \xrightarrow{\Phi_{X,Y}}
        \coprod_{p\in \FuncT(X)}T(Y)^{\supp(p)}
        \xrightarrow{\coprod_{p\in \FuncT(X)}{\eta_{T(Y)}}^{\supp(p)}}
        \coprod_{p\in \FuncT(X)}
        \FuncT(T(Y))^{\supp(p)}
        \xrightarrow{\Phi^{-1}_{X,T(Y)}}
        \FuncT(X\times \FuncT(Y)).
    \end{equation}
    In the case of the distribution monad, which is a monad with support as described in detail below, this recovers \cite{jacobs2017hyper}'s hyper normalisation map.
    Our monads with support could thus be seen as a generalisation of hyper normalisation.
    
    On the other hand, \cite{virgoUnifilarMachinesAdjoint2023}'s version of unifilarisation is expressed in terms of strongly representable Markov categories \cite{fritz2020-representable}, which are a different framework from monads with support.
    In particular, an example of a monad generating a strongly representable Markov category is the Giry monad on standard Borel spaces, which doesn't seem to be a monad with support because there doesn't seem to be a useful notion of support that applies to all morphisms in $\mathbf{BorelStoch}$; see theorem 3.2.7 of \cite{fritz2023absolute}, for example.
    
    A hyper normalisation-like map can be constructed in any strongly representable Markov category as well.
    In the notation of \cite{fritz2020-representable}, we apply the defining property of a strongly representable Markov category to the sampling map $\mathsf{samp}_{X\otimes Y}$ to get a stochastic map $\mathsf{samp}_{X\otimes Y}^{\sharp Y}:P(X\otimes Y)\to X\otimes P(Y)$, and then take its deterministic version to get a deterministic map $(\mathsf{samp}_{X\otimes Y}^{\sharp Y})^\sharp:P(X\otimes Y)\to P(X\otimes P(Y))$, which has the same type as the hyper normalisation map.
    In the case of the Kleisli category of the distribution monad, this also recovers Jacobs' construction.
    
    Thus, monads with support and strongly representable Markov categories can both be seen as generalisations of hyper normalisation.
    We don't currently know whether strongly representable Markov categories are strictly more general than monads with support, or whether the two might admit a common generalisation.
    We also don't know how these observations relate to other generalisations of hyper normalisation found in \cite{bohinen2025categorical} and \cite{GARNER2023105044}.
\end{remark}

\subsection{Examples of monads with support}

Our first example is the main motivating example for the concept.

\begin{example}
    The distribution monad $\Dist$ is a monad with support.
\end{example}
To see this, we give some definitions from elementary probability theory, in the language of the distribution monad.
In the following, given $p \in \Dist(X)$ and $x \in X$, we write $p(x)$ for the probability assigned to $x$ according to the distribution $p$, and similarly we write $p'(x,y)$ for the probabilities assigned by a distribution $p' \in \Dist(X \times Y)$.

\begin{definition}[Support of a probability distribution]
    Let $\Dist$ be the distribution monad. 
    Given a set $X$, let the map $\supp_X : \Dist(X) \to \Set$ be given by $\supp_X(p) = \{x \in X \mid p(x) > 0\}$.
    Since there is little danger of ambiguity we will just write this as $\supp(p)$.
\end{definition}

\begin{definition}[Marginal probability]
    \label{def:marginalProbability}
    Given an element $p$ of $\Dist(X \times Y)$ (seen as a joint distribution between $X$ and $Y$), define the \emph{marginal over $X$} as $p_X \in \Dist(X)$, where $p_X(x) = \sum_y p(x,y)$.
\end{definition}
This is well defined because since $p$ is finitely supported then there are only finitely many $x \in X$ where $p_X(x)$ is nonzero, so $p_X$ is finitely supported as well.
We now define conditional probability distributions, restricting their domain in order to avoid dividing by zero.
\begin{definition}[Conditional probability]
    \label{def:conditionalProbability}
    Given $p \in \Dist(X \times Y)$ and $x \in \supp{p_X}$, define the \emph{conditional distribution over $Y$ given $x$} as $p_{Y \mid x} \in \Dist(Y)$, where $p_{Y \mid x}(y) = p(x,y)/p_X(x)$.
\end{definition}
We can again note that this is well defined, because since $p$ is finitely supported, there are only finitely many $y$ for which $p_{Y \mid x}(y)>0$.
We can then express a basic fact about probability:
\begin{lemma}
    \label{thm:disintegration}
    Given sets $X$ and $Y$, there is an isomorphism
    \begin{align}
        \Phi_{X, Y} : \Dist(X \times Y) \xrightarrow{\cong} \coprod_{p \in \Dist(X)} \Dist(Y)^{\supp(p)}
    \end{align}
    natural in $Y$.
    In the $\Dist(X \times Y) \to \coprod_{p \in \Dist(X)} \Dist(Y)^{\supp(p)}$ direction, this is given by ``disintegrating'' the distribution into a marginal and a conditional:
    \begin{align}
        \label{eqn:disintegration}
        p \mapsto (
            p_X
            \,,\,
            \lambda x \,.\, p_{Y \mid x}
        ).
    \end{align}
    The inverse in the $\coprod_{p \in \Dist(X)} \Dist(Y)^{\supp(p)}\to \Dist(X \times Y)$ direction maps $(
        q_X
        \,,\,
        \lambda x \,.\, q_{Y\mid x}
    )
    $
    to $q \in \Dist(X \times Y)$, where $q(x,y) \coloneqq q_X(x)q_{Y\mid x}(y)$.
\end{lemma}

\begin{proof}
    The disintegration of a joint distribution into marginal and conditional is a well known feature of elementary probability theory; see example 3.8 of \cite{choDisintegrationBayesianInversion2019} for example. 
    Its application here is straightforward.
\end{proof}
In the following, $\PowerFinite^+$ refers to the nonempty powerset monad, which is the same as the usual finite powerset monad $\PowerFinite$, except that we exclude the empty set, so that $\PowerFinite^+(X) = \{ U \subseteq X \mid U \text{ is finite and } U \ne \varnothing \}$.

\begin{example}
    The finite nonempty powerset monad $\PowerFinite^+$ is a monad with support, where $\supp(U) = U$.
\end{example}

That is, since $U\in \PowerFinite^+$ is already a set, $U\subseteq X$, we can take the support to be $U$ itself.
The isomorphism 
\begin{align}
    \Phi_{X, Y} : \PowerFinite^+(X \times Y) \xrightarrow{\cong} \coprod_{A \in \PowerFinite^+(X)} \PowerFinite^+(Y)^{\supp(A)} .
\end{align}
is then given in the forward direction by 
\begin{align}
    \label{eqn:finite-nonempty-powerset}
    U \mapsto \big(
        \{x\in X \mid \exists y \in Y. (x,y)\in U\}
        \,,\,
        \lambda x \,.\, \{y\in Y\mid (x,y)\in U\}
    \big)
\end{align}
and in the reverse direction by
\begin{align}
    (U_X,f)\mapsto \{(x,y) \in X \times Y \mid x \in U_X,\ y \in f(x)\}.
\end{align} 

\begin{remark}
    Another example of a monad with support is the nonempty powerset monad $\Power^+$, since nothing above requires the finiteness condition.
    However, like the powerset monad, the nonempty powerset monad lacks a final coalgebra.
\end{remark}

\section{Unifilarisation as determinisation of Mealy machines}

\subsection{Mealy machines}
As hinted at before, we now formally switch from Moore machines to Mealy machines because unifilarisation is naturally formulated in terms of one-step joint behaviour of outputs and next states. 
This is the key structural difference from the Moore case: in a Mealy coalgebra $\States \to \FuncT(\Observations \times \States)^\Inputs$, the output and next state are packaged together, so after averaging one can rewrite the resulting joint distribution as an output marginal together with posterior next-state distributions.
That rewriting step is what later produces Bayesian updates.
We thus initially introduce Mealy machines as coalgebras following~\cite{ruttenUniversalCoalgebraTheory2000, jacobsIntroductionCoalgebraMathematics2017}.
To capture existing definitions that include deterministic, nondeterministic, and stochastic Mealy machines~\cite{virgoUnifilarMachinesAdjoint2023, bonchiEffectfulMealyMachines2025a}, we parametrise our definition with a monad $\FuncT$.
\begin{definition}[Mealy machines]
    \label{def:mealyMachines}
    Mealy machines are coalgebras of the functor
    \begin{align}
        \FuncF_\structmealy({-}) = \FuncT(\Observations \times {-})^\Inputs.
    \end{align}
\end{definition}
Using monads with support~\cref{def:monadSupport}, we then introduce the following notion.

\begin{definition}[Supported Mealy machines]
    \label{def:mealyWithSupport}
    Given a monad with support $\FuncT$, supported Mealy machines are coalgebras of the functor
    \begin{align}
        \FuncF_\suppmealy({-}) = \left(\coprod_{p \in \FuncT(\Observations)} \FuncT({-})^{\supp(p)}\right)^\Inputs.
    \end{align}
\end{definition}
A supported Mealy machine can be seen as a gadget that starts in some state in $X$, receives an input in $\Inputs$, and then produces both an output in $\Observations$ and its own next state in $X$ under the branching specified by the monad $\FuncT$.
In general the output and the next state can be correlated, given the previous state and the input.
Calling these ``Mealy machines'' is justified by \cref{eqn:monad-with-support-iso}, which implies that $\FuncF_\suppmealy\cong \FuncF_\structmealy = \FuncT(O\times{-})^\Inputs$.

\begin{remark}
    Since we will only consider supported Mealy machines from now on, we shall refer to supported Mealy machines simply as Mealy machines.
\end{remark}
An important example of Mealy machines is that of stochastic Mealy machines.
\begin{definition}[Stochastic Mealy machines]
    Stochastic Mealy machines are coalgebras of the functor $\FuncF_\stochsuppmealy$, i.e. Mealy machines for $\FuncT = \Dist$.
\end{definition}
Stochastic Mealy machines correspond to partially observable Markov decision processes (POMDPs) in control theory and reinforcement learning. \footnote{We note that in the reinforcement learning literature it is also common to formulate POMDPs as stochastic Moore machines.}
In reinforcement learning, POMDPs usually also include also a map to represent rewards but in traditional control theory this is not necessarily the case~\cite{astromOptimalControlMarkov1965}.

Next we show how unifilarisation can be expressed for Mealy machines and in particular for stochastic ones, which capture the definition by~\cite{virgoUnifilarMachinesAdjoint2023}.

\subsection{Unifilarisation of Mealy machines}

Here we adapt the unifilarisation definition given in~\cite{virgoUnifilarMachinesAdjoint2023} to our framework, in which Mealy machines are seen as coalgebras.
To do that, we first introduce \emph{unifilar} Mealy machines.
These are a variation on the concept of a Mealy machine, which is somewhat less intuitive but is closely related to a central concept in computational mechanics, that of a unifilar hidden Markov process~\cite{shaliziComputationalMechanicsPattern2001, barnettComputationalMechanicsInput2015}.
\begin{definition}[Unifilar Mealy machines]
    \label{def:unifilarMealyMachines}
    Unifilar Mealy machines are coalgebras of the (polynomial) functor
    \begin{align}
        \label{eqn:unifilar-machines-functor}
        \FuncF_\structunifmealy({-}) = \left(\coprod_{p \in \FuncT(\Observations)} ({-})^{\supp(p)}\right)^\Inputs.
    \end{align}
\end{definition}
A unifilar Mealy machine can be seen as a gadget that starts in some state in $X$, takes an input in $\Inputs$, generates an output in $\Observations$ with branching specified by $\FuncT$, and then \emph{deterministically} updates its state, as a function of both the given input and its own generated output, which can be stochastic, nondeterministic, etc.
Its next state is only defined on outputs that it produces with positive probability, given its current state and input.

We will show shortly that, for $\FuncT = \Dist$, Mealy machines can be given a $\FuncT$-algebra structure such that the determinisation of a Mealy machine is a unifilar Mealy machine.
Before we do that, however, we first spell out what the definition of unifilarisation from~\cite{virgoUnifilarMachinesAdjoint2023} amounts to when specialised to the distribution monad, in our current notation.
We will later show that determinisation produces the same result.

\begin{definition}[Unifilarisation of stochastic Mealy machines~\cite{virgoUnifilarMachinesAdjoint2023}]
    \label{def:unifilarisation}
    Let $(\States, \transitionCoalg_\stochsuppmealy)$
    be a stochastic Mealy machine. 
    Its unifilarisation is the unifilar Mealy machine $(\Dist(\States), \transitionCoalg_\stochunifmealy)$ where transitions
    \begin{align}
        \transitionCoalg_\stochunifmealy : \Dist(\States) \to \left(\coprod_{p \in \Dist(\Observations)} \Dist(\States)^{\supp(p)}\right)^\Inputs
    \end{align}
    are defined as follows.
    For $q \in \Dist(\States)$ and $\inputs \in \Inputs$, first define the joint distribution
    \begin{align}
        \label{eq:unifilarisation-joint}
        \rho_{q, \inputs}(\observations, \states') \coloneqq
        \sum_{s \in \States} q(\states)\,\Phi_{\Observations,\States}^{-1} \left(\transitionCoalg_\stochsuppmealy(\states)(\inputs) \right)(o, \states').
    \end{align}
    This takes a little unpacking.
    The idea is that for each state $\states$, $\transitionCoalg_\stochsuppmealy(\states)(\inputs)$ is an element of $\left(\coprod_{p \in \Dist(\Observations)} \FuncT(\States)^{\supp(p)}\right)$, and so we can turn it into a joint distribution over $\Observations \times \States$ via the isomorphism $\Phi_{\Observations, \States}^{-1}$.
    We then take the average of all of these.
    Conceptually, all we are doing here is calculating the expected distribution over outputs $s', o$ of the Mealy machine, given an input $i$ and a distribution $q$ over initial states.
    
    Next we define the output marginal over $\Observations$ and, for each $\observations \in \supp(p_{q, \inputs})$, the posterior next-state distribution:
    \begin{align}
        \label{eqn:explicit-unifilarisation}
        p_{q, \inputs}(\observations) \coloneqq\sum_{\states' \in \States}\rho_{q, \inputs}(\observations, \states') 
        \qquad
        k_{q, \inputs}(\observations)(\states') \coloneqq\frac{\rho_{q, \inputs}(\observations, \states')}{p_{q, \inputs}(\observations)}.
    \end{align}
    We then set $\transitionCoalg_\stochunifmealy(q)(\inputs) \coloneqq (p_{q, \inputs}, k_{q, \inputs})$.
\end{definition}

Unifilarisation produces an epistemic model, a unifilar Mealy machine, from a given stochastic Mealy machine that could (but doesn't have to) represent an environment for an agent.
It is called ``epistemic'' because the state space of the unifilar Mealy machine is given by \emph{beliefs} about $X$, $\Dist(X)$, that could be said to belong to an idealised Bayesian reasoner, or to an agent.
The idea is that the reasoner's prior at any given time is an element of the state space $\Dist(X)$, which is updated following Bayesian filtering.
This process underlies the construction of belief MDPs in reinforcement learning~\cite{kaelblingPlanningActingPartially1998} and their previous formulations in terms of state information and sufficient statistics for control in control theory~\cite{astromOptimalControlMarkov1965, striebelSufficientStatisticsOptimum1965}, which correspond to unifilar Mealy machines with marginalised observations~\cite{virgoUnifilarMachinesAdjoint2023}.
It also generalises the construction of mixed state presentations in computational mechanics~\cite{ellisonPredictionRetrodictionAmount2009, crutchfieldExactComplexitySpectral2016}.
For $\FuncT = \PowerFinite^+$, we obtain the same construction proposed by~\cite{virgoUnifilarMachinesAdjoint2023} since the Kleisli category of the (finite) nonempty powerset monad is also a strongly representable Markov category.

Next we can note that $\FuncF_\suppmealy = \FuncF_\structunifmealy\FuncT$, so that (supported) Mealy machines are coalgebras of $\FuncF_\structunifmealy \FuncT$.
After showing that coalgebras of this type carry a $\FuncT$-algebra structure, we will then prove that the determinisation construction applied to these coalgebras will give unifilar machines implementing Bayesian updates.

\subsection{Determinisation of Mealy machines}
\label{sec:detMealyMachines}

Mealy machines can be equipped with a $\FuncT$-algebra structure.
To do so, we first give $\coprod_{t\in \FuncT(\Observations)}\FuncT(\States)^{\supp(t)}$ an algebra structure.

\begin{proposition}
    The map $\nu : \FuncT \left( \coprod_{t \in \FuncT(\Observations)} \FuncT(\States)^{\supp(t)} \right)\to
        \coprod_{t \in \FuncT(\Observations)} \FuncT(\States)^{\supp(t)}$,
        given by
    \begin{align}
        \nu \coloneqq 
        \FuncT \left( \coprod_{t\in \FuncT(\Observations)}\FuncT(\States)^{\supp(t)}
        \right)
        \xrightarrow{\FuncT(\Phi^{-1}_{\Observations, \States})}
        \FuncT \left( \FuncT(\Observations \times \States) \right)
        \xrightarrow{\mult_{\Observations \times \States}}
        \FuncT(\Observations \times \States)
        \xrightarrow{\Phi_{\Observations, \States}}
        \coprod_{t \in \FuncT(\Observations)}\FuncT(\States)^{\supp(t)},
    \end{align}
    induces an algebra structure on $\coprod_{t\in \FuncT(\Observations)}\FuncT(\States)^{\supp(t)}$.
\end{proposition}
\begin{proof}
    The equations defining an algebra can be verified straightforwardly.
\end{proof}
Then $\FuncF_\suppmealy(S) = \left(\coprod_{t\in \FuncT(\Observations)}\FuncT(\States)^{\supp(t)}
\right)^\Inputs$ has an algebra structure given by
\begin{align}
    \label{eq:algebraMorphStructMealy}
    \xi \coloneqq \FuncT\left(\left(\coprod_{t\in \FuncT(\Observations)}\FuncT(\States)^{\supp(t)}
    \right)^\Inputs\right)
    \xrightarrow{\nu^{(\Inputs)}}   \left(\coprod_{t\in \FuncT(\Observations)}\FuncT(\States)^{\supp(t)}
    \right)^\Inputs, 
\end{align}
which is an algebra structure by \cref{prop:algebra-on-exponental}.

Since Mealy machines with this $\FuncT$-algebra structure are $\FuncF \FuncT$-algebras, they can be determinised using~\cite{silvaGeneralizingDeterminizationAutomata2013}'s construction, which yields a determinised machine 
\begin{align}
    \label{eq:detMachineUnif}
    (\FuncT(\States), \transitionCoalg_\suppmealy^\sharp = \xi \circ \FuncT \transitionCoalg_\suppmealy)
\end{align}
with $\xi$ as in \cref{eq:algebraMorphStructMealy}.

As our main result, we now show that in the case $\FuncT=\Dist$, these determinised machines are exactly the ones given by the unifilarisation of~\cite{virgoUnifilarMachinesAdjoint2023} in the form of \cref{def:unifilarisation}.

\begin{theorem}[Unifilarisation is determinisation]
    Let $(\States, \transitionCoalg_\stochsuppmealy : \States \to \left(\coprod_{p \in \Dist(\Observations)} \Dist(\States)^{\supp(p)}\right)^\Inputs$ be a stochastic Mealy machine, equipped with the algebra map in \cref{eq:algebraMorphStructMealy}, with $\FuncT=\Dist$.
    Its determinisation is exactly its unifilarisation according to \cref{def:unifilarisation}.
\end{theorem}

\begin{proof}[Sketch proof]
    The determinisation of $(\States, \transitionCoalg_\stochsuppmealy)$ is given by $(\Dist(\States),\transitionCoalg^\sharp_\stochsuppmealy)$, where, using~\cref{eq:algebraMorphStructMealy,eq:detMachineUnif},
    \begin{align}
        \transitionCoalg_\stochsuppmealy^\sharp =  \Dist(\States)\xrightarrow{\Dist(f_\stochsuppmealy)}\Dist\left(\left(\coprod_{p\in \Dist(\Observations)}\Dist(\States)^{\supp(p)}\right)^\Inputs\right)
        \;\xrightarrow{\nu^{(\Inputs)}}\;
        \left(\coprod_{p\in \Dist(\Observations)}\Dist(\States)^{\supp(p)}\right)^\Inputs.
    \end{align}
    We wish to show that, for $q \in \Dist(\States)$ and $\inputs \in \Inputs$, we have $\transitionCoalg^\sharp_\stochmealy(q)(\inputs) = (p_{q, \inputs},k_{q, \inputs})$, with $p_{q, \inputs}$ and $k_{q, \inputs}$ as defined in \cref{eqn:explicit-unifilarisation}.
    
    To proceed, we uncurry $\nu^{(\Inputs)}$ and expand the definition of $\nu$ noting that $\transitionCoalg_\stochsuppmealy^\sharp$ is the curried version of
    \begin{multline}
        \label{eq:long-chain-2}
        \Dist(\States)\times \Inputs \xrightarrow{\Dist(f_\stochsuppmealy)\times \id_\Inputs}
        \Dist\left(\left(\coprod_{t\in \Dist(\Observations)}\Dist(\States)^{\supp(t)}\right)^\Inputs\right)\times \Inputs
        \xrightarrow{\mathrm{str}_{(\cdot)^\Inputs, \Inputs}}
        \Dist\left(\left(\coprod_{t\in \Dist(\Observations)}\Dist(\States)^{\supp(t)}\right)^\Inputs\times \Inputs\right)\\
        \xrightarrow{\Dist(\mathrm{eval}_{(\cdot), \Inputs})}
        \Dist\left(\coprod_{t\in \Dist(\Observations)}\Dist(\States)^{\supp(t)}\right)
        \xrightarrow{\Dist(\Phi^{-1}_{\Observations, \States})}
        \Dist \left( \Dist(\Observations \times \States) \right)
        \xrightarrow{\mult_{\Observations \times \States}}
        \Dist(\Observations \times \States)
        \xrightarrow{\Phi_{\Observations, \States}}
    \coprod_{t\in \Dist(\Observations)}\Dist(\States)^{\supp(t)}.
    \end{multline}

    The first five steps of \cref{eq:long-chain-2} are performing the same calculation as \cref{eq:unifilarisation-joint}, taking as input $q\in DS$ and $i\in I$ and forming the distribution $\rho_{q, \inputs}(o,s')$.
    The last step forms the marginal $p_{q, \inputs}(\observations)$ and the conditional $k_{q, \inputs}(\observations)(\states')$, as in \cref{eqn:explicit-unifilarisation}, via the isomorphism in \cref{thm:disintegration}.

\end{proof}

As with stochastic Moore machines in \cref{sec:structured-moore-machines-example}, it is helpful to compare this to the final coalgebra semantics.
The final coalgebra for unifilar machines is worked out in~\cite{virgoUnifilarMachinesAdjoint2023}. We summarise it here, starting with the following definition.

\begin{definition}[Causal stochastic behaviour]
    Let $\Dist$ be the distribution monad, and let $X, Y \in \Set$.
    For each $n \ge 1$, write $X^n$ and $Y^n$ for the $n$-fold Cartesian products,
    and let
    \begin{align}
        \pi_n^X : X^{n+1}\to X^n,
        \qquad
        \pi_n^Y : Y^{n+1}\to Y^n
    \end{align}
    be the prefix projections
    given by $\pi_n^X(x_0, \dots, x_n) = (x_0, \dots, x_{n-1})$, $\pi_n^Y(y_0, \dots, y_n) = (y_0, \dots, y_{n-1})$.
    A causal stochastic behaviour $b$ from $X$ to $Y$ is a family
    \begin{align}
        b_n : X^n \to \Dist(Y^n)
        \qquad (n \ge 1)
    \end{align}
    such that
    \begin{align}
        \Dist(\pi_n^Y) \circ b_{n+1} = b_n \circ \pi_n^X 
        \qquad
        \forall n \ge 1
    \end{align}
    We write $\Caus_\Dist(X, Y)$ for the set of all such families.
\end{definition}

These are called ``causal'' because each distribution of the first $n$ output symbols depends only on the first $n$ input symbols.
Intuitively, this implies that an output can only depend on inputs that occur in its past.

If we have some $p\in D(Y^n)\cong D(Y\times Y^{n-1})$, we can regard it as a joint distribution between $Y$ (playing the role of the first element in a list) and $Y^{n-1}$ (the tail of the list).
We can then use marginal and conditional probability (\cref{def:marginalProbability,def:conditionalProbability}), writing $p_{Y^{n-1}\mid y}$ for the conditional distribution over the tail of the list, given that its first element was $y$.

We will need an operation called \emph{evolution} \cite{mcgregorFormalisingIntentionalStance2025a}: given a causal stochastic behaviour $b\in \Caus_\Dist(X,Y)$, along with $x\in X$ and $y\in \supp (b_1(x))$, 
we write $b\bullet(x,y) \in \Caus_\Dist(X,Y)$ (read as ``$b$ evolved by $x$ and $y$'') for the causal stochastic behaviour where
\begin{equation}
\big(b\bullet (x,y)\big)_k (x_1,\dots, x_k)
=
\big(
b_{k+1}(x,x_1,\dots, x_k)
\big)_{Y^{k-1}\mid y}\,.
\end{equation}
Intuitively, to obtain $b\bullet (x,y)$, we first feed in the input $x$ and then condition on the output being $y$; this then gives us a new causal stochastic behaviour.

Given input and output sets $I,O$, the set $\Caus_\Dist(I,O)$ forms the carrier of a coalgebra
    \begin{align}
    \label{eqn:finalCausalBehaviour}
        \left( \Caus_\Dist(\Inputs, \Observations), \transitionCoalgFinal_\stochmealy^\sharp : \Caus_\Dist(\Inputs, \Observations) \to \left( \sum_{p \in \Dist(\Observations)} \left( \Caus_\Dist(\Inputs, \Observations) \right)^{\supp(p)} \right)^\Inputs \right),
    \end{align}
where $\transitionCoalgFinal_\stochmealy^\sharp$ is defined as
\begin{equation}
\transitionCoalgFinal_\stochmealy^\sharp(b)
    =
    \lambda\, i:I\,.\,
    \big(b_1(i),
    \lambda\, o:\supp(b_1(i))\,.\,
    b\bullet(i,o)
    \big).
\end{equation}
We can now state the following proposition:

\begin{proposition}[Final coalgebra of determinised stochastic Mealy machines~\cite{virgoUnifilarMachinesAdjoint2023}]
    The endofunctor $\left( \coprod_{p \in \Dist(\Observations)} ({-})^{\supp(p)} \right)^\Inputs$ has final coalgebra given by \cref{eqn:finalCausalBehaviour}.
\end{proposition}
\begin{proof}
    The proof is essentially the same as for \cite[Proposition 3.2]{virgoUnifilarMachinesAdjoint2023}, although some slight adaptations have to be made due to our focus on Mealy machines rather than the Moore-like comb machines of \cite{virgoUnifilarMachinesAdjoint2023}.
\end{proof}

Unlike the case of structured Moore machines, this semantics records full causal input-output behaviour,
rather than only the last output distribution after a finite input word.
This gives some intuition for why Bayesian conditioning appears in the determinisation of supported Mealy machines but not for stochastic Moore machines:
in order to correctly predict future outputs, the determinised machine needs to take account of what the previous output was, and Bayesian conditioning is the natural way to achieve this.

\section{Conclusion}
The main result of this paper is that unifilarisation fits the general pattern of coalgebraic determinisation. The point is not that it coincides with the standard Moore-style construction, but that it arises from the same determinisation principle once the transition type is changed appropriately. 
By moving from Moore machines to Mealy machines over monads with support, we obtain a determinisation procedure on $\FuncT(X)$ that, in the case of the distribution monad, reproduces exactly the Bayesian filtering update of stochastic unifilar machines.

This also clarifies the source of the semantic gap between the two settings. 
Standard determinisation of stochastic Moore machines keeps track only of the distribution of the current output after a finite input history, in line with the idea of language acceptance from automata theory. 
Unifilarisation, by contrast, keeps track of output histories together with the posterior updates they induce, and therefore yields a semantics in terms of causal stochastic behaviours, which is more in line with agents solving problems such as POMDPs in reinforcement learning. 
The contribution of the paper is therefore twofold: it identifies unifilarisation as an instance of coalgebraic determinisation, and it explains precisely why this instance induces a richer behavioural semantics than the standard stochastic Moore case.

\bibliographystyle{eptcsalpha}
\bibliography{AllEntriesZotero, MoreRefs}

\end{document}